\documentclass[a4paper,11pt]{amsart} 
\usepackage{amsmath,amsxtra,amssymb,latexsym, amscd,amsthm}
\usepackage[mathscr]{eucal}
\usepackage{mathrsfs}
\usepackage{bbm}
\usepackage{enumerate}
\usepackage{pict2e}
\usepackage{tikz}
\usepackage{graphicx}
\usepackage{subfigure}
\usepackage{subcaption}
\usepackage[a4paper]{geometry}
\usepackage{layouts}
\usepackage{color}
\usepackage{hyperref}
\numberwithin{equation}{section}

\theoremstyle{plain}
\newtheorem{thm}{Theorem}[section]

\newtheorem{prp}[thm]{Proposition}
\newtheorem{cor}[thm]{Corollary}

\newtheorem*{ego*}{Egorov's Theorem}
\theoremstyle{definition}

\newtheorem*{rem*}{Remark}

\newcommand{\dd}{\mathrm{d}}
\newcommand{\ee}{\mathrm{e}}
\newcommand{\ii}{\mathrm{i}}
\newcommand{\Z}{\mathbb{Z}}

\newcommand{\R}{\mathbb{R}}

\newcommand{\T}{\mathbb{T}}

\newcommand{\LL}{\mathbb{L}}

\newcommand{\cA}{\mathcal{A}}

\newcommand{\cC}{\mathcal{C}}

\newcommand{\cH}{\mathcal{H}}

\newcommand{\cG}{\mathcal{G}}
\newcommand{\fa}{\mathfrak{a}}
\newcommand{\fb}{\mathfrak{b}}
\newcommand{\bn}{\mathbf{n}}
\newcommand{\bm}{\mathbf{m}}
\newcommand{\bk}{\mathbf{k}}
\newcommand{\br}{\mathbf{r}}

\DeclareSymbolFont{extraup}{U}{zavm}{m}{n}
\DeclareMathSymbol{\varheart}{\mathalpha}{extraup}{86}
\DeclareMathSymbol{\vardiamond}{\mathalpha}{extraup}{87}
\makeatletter
\newcommand*{\ovF}[1]{%
  $\m@th\overline{\raisebox{0pt}[\dimexpr\height+0.3mm\relax]{#1}}$%
	}
\makeatother
\title[Limiting distributions of ergodic quantum walks]{Limiting distributions of ergodic continuous-time quantum walks on periodic graphs}
\author{Anne Boutet de Monvel, Kiran Kumar A.S., Mostafa Sabri}
\address{Institut de Math\'ematiques de Jussieu-Paris Rive Gauche, Universit\'e de Paris, 8 place Aur\'elie Nemours, case 7012, 75205 Paris Cedex 13, France.}
\email{anne.boutet-de-monvel@imj-prg.fr}
\address{Science Division, New York University Abu Dhabi, Saadiyat Island, Abu Dhabi, UAE.}
\email{kas10285@nyu.edu}
\email{mostafa.sabri@nyu.edu}
\usepackage{calc}
\usepackage{graphicx}

\makeatletter
\newlength{\temp@wc@width}
\newlength{\temp@wc@height}
\newcommand{\widecheck}[1]{%
  \setlength{\temp@wc@width}{\widthof{$#1$}}%
  \setlength{\temp@wc@height}{\heightof{$#1$}}%
  #1\hspace{-\temp@wc@width}%
  \raisebox{\temp@wc@height+2pt}[\heightof{$\widehat{#1}$}]%
     {\rotatebox[origin=c]{180}{\vbox to 0pt{\hbox{$\widehat{\hphantom{#1}}$}}}}%
}
\makeatother

\begin{document}

\begin{abstract}
In this expository note, we study several families of periodic graphs which satisfy a sufficient condition for the ergodicity of the associated continuous-time quantum walk. For these graphs, we compute the limiting distribution of the walk explicitly. We uncover interesting behavior where in some families, the walk is ergodic in both horizontal and sectional directions, while in others, ergodicity only holds in the horizontal (large $N$) direction. We compare this to the limiting distribution of classical random walks on the same graphs.
\end{abstract}

\maketitle

\section{Introduction}

This paper is concerned with studying the quantum dynamics on crystals $\Gamma$, i.e. connected $\Z^d$-periodic graphs. The vertex set takes the form
\begin{equation}\label{e:v=vf}
V = V_f + \Z_\fa^d \,,
\end{equation}
where
\[
V_f = \{v_1,v_2,\dots,v_\nu\}
\]
is a finite set representing the vertices of a \emph{fundamental domain} and $\Z_\fa^d$ is the lattice $\Z^d$ expressed in some coordinate system $\fa_1,\dots,\fa_d$, i.e. $(\fa_j)$ are linearly independent vectors in $\R^d$. If $\fa_i = \mathfrak{e}_i$ is the standard basis, $\Z_\fa^d=\Z^d$. In general, $\Z_\fa^d:=\{\sum_{i=1}^d n_i\fa_i: \bn\in \Z^d\}$. 

For example, $\Gamma=\Z^d$ has $V_f=\{0\}$ and $\fa_i=\mathfrak{e}_i$. For an infinite $1d$ $k$-strip in $\Z^2$, we can take $v_i=(0,i)$, $i=0,\dots,k$, here $d=1$ and $\fa_1=1$. This example illustrates that $\Gamma$ is in general embedded in a Euclidean space which may have dimension larger than $d$. If $\Gamma$ is the honeycomb lattice, $V_f=\{(0,0),v\}$ with $v=\frac{a}{2}(1,\frac{1}{\sqrt{3}})$, $\fa_1=a(1,0)$, $\fa_2=\frac{a}{2}(1,\sqrt{3})$, where $a>0$.

By \eqref{e:v=vf}, any $v\in V$ takes the form $v=v_i+\bn_\fa$. For all practical purposes we may identify $\ell^2(\Gamma)\equiv \ell^2(\Z^d)^\nu$ via 
\begin{equation}\label{e:identi}
\psi(v_i+\bn_\fa) \leftrightarrow \psi_i(\bn)
\end{equation}
for $i=1,\dots,\nu$ and $\bn\in\Z^d$. For example, functions on a $k$-strip are identified with vector functions on $\Z$, each vector having length $k$.

The adjacency matrix on $\Gamma$ is defined as usual by $(\cA f)(v) = \sum_{w\sim v} f(w)$ for $v\in \Gamma$, where $v\sim w$ means that $v$ and $w$ are neighbors. The \emph{continuous-time quantum walk on $\Gamma$} is given by the unitary semigroup $\ee^{\ii t\cA}$. In this paper we are concerned with lattices $\Gamma$ whose quantum dynamics is \emph{ergodic}. In analogy to the classical setting where a walk starts from some vertex $v\in \Gamma$ and hops at random to its neighbors at each step, here we start from an initial state $\delta_v$, where $\delta_v(w)=1$ if $w=v$ and $0$ otherwise, and study the long-time behavior of $\ee^{\ii t\cA}\delta_v$.

More precisely, we consider finite subsets $\Gamma_N\subset \Gamma$ given by $\Gamma_N = \cup_{n\in \LL_N^d}(V_f+\bn_\fa)$, where $\LL_N^d=\{0,\dots,N\}^d$, and consider the restriction $A_N$ of $\cA$ to $\Gamma_N$ with \emph{periodic boundary conditions}. For example, if $\Gamma=\Z$, then we are approximating the graph with $N$-cycles. For a fixed $v\in \Gamma_N$, we then consider $\mu_{T,v}^N(w):= \frac{1}{T}\int_0^T |(\ee^{\ii t A_N}\delta_v)(w)|^2\,\dd t$, which is a probability measure on $\Gamma_N$ as $\ee^{\ii t A_N}$ is unitary, and we compare it with the \emph{uniform measure} on $\Gamma_N$, that is $\mu^N(w)=\frac{1}{|\Gamma_N|}$. A quantum walk is ergodic if $\mu_{T,v}^N\approx \mu^N$ when $T\to\infty$, followed by $N\to\infty$. See \cite{BdMS} for more details on why one should consider the mean time average and this particular order of limits. More broadly, we say the quantum walk is ergodic if $\mu_{T,v}^N$ approaches a measure $\tilde{\mu}^N_v$ which ``has a density'' with respect to the uniform measure $\mu^N$. We call $\tilde{\mu}^N_v$ the \emph{limiting distribution} of the walk, and computing it on simple examples, when it exists, is the main purpose of the present paper.

Ergodicity is a strong and very precise result of \emph{delocalization}. Let us discuss this connection a bit further. In general the spectrum of $\cA$ on the periodic graph $\Gamma$ consists of $m\le \nu$ bands of absolutely continuous spectrum, and possibly a finite number of eigenvalues which are called \emph{flat bands}, see \cite[Section 2]{SY} and references therein. There is no singularly continuous spectrum. It is easy to show that $\mu_{T,\psi}(w):=\frac{1}{T}\int_0^T |(\ee^{\ii t\cA}\psi)(w)|^2\,\dd t \to \sum_k |(P_k\psi)(w)|^2$ as $T\to\infty$, where $P_k$ is the orthogonal projection onto the distinct eigenvalues of $\cA$. In particular, in the regime of spectral delocalization where $\cA$ has no eigenvalues, one gets that $\mu_{T,v}(w):=\frac{1}{T}\int_0^T |(\ee^{\ii t\cA}\delta_v)(w)|^2\,\dd t$ converges to zero everywhere. Ergodicity implies more strongly that this decay to zero is approximately \emph{the same at all vertices} when considering an approximating sublattice $\Gamma_N$, namely $\mu_{T,v}^N\approx \frac{1}{|\Gamma_N|}$ (which goes to zero when $N\to\infty$). As we mentioned earlier, in general the limiting measure $\tilde{\mu}^N_v$ may have a density, so that the mass on each vertex $w$ is some $\frac{f(v,w)}{|\Gamma_N|}$.

Let us mention that, when working directly on the infinite graph $\Gamma$, since $\mu_{T,v}\xrightarrow{w} 0$ in the absence of eigenvalues, one can be interested in studying the speed of decay and properly normalizing the process to obtain a nontrivial limit. There is a large literature on this topic, known as \emph{limit theorems} in the quantum walks community. In this framework it is not necessary to consider time averages, one can study $m_{t,\psi}(w):=|(\ee^{\ii t\cA}\psi)(w)|^2$ directly. Then the correct scaling is linear (the motion is ballistic), and if we study the process per unit time (i.e. study the limit of the random variable $\frac{X_t}{t}$, where $X_t$ has distribution $m_{t,\psi}$), then the limit is nontrivial. We refer to \cite{Ko} for one of the earliest results in the case of discrete-time quantum walks on $\Z$. As this is not the topic of this paper, for further extensions we only mention \cite{GJS} for a multi-dimensional analog, \cite{HKSS} for certain results on crystals and \cite[Section 5]{BdMS0} for limits theorems of continuous-time quantum walks on periodic graphs $\Gamma$. See e.g. \cite{Ko0} for a more comprehensive survey of this topic.

\subsection{The ``black box'' theorem} 

In an earlier paper \cite{BdMS} we proved a general criterion for the ergodicity of continuous-time quantum walks. This result holds more generally for unitary semigroups $\ee^{-\ii t \cH}$ where $\cH$ is a periodic Schr\"orindger operator, $\cH = \cA+Q$. To state the result, let us recall that using Bloch-Floquet theory, $\cH$ is unitarily equivalent to a multiplication operator $M_H$ on $\ell^2(\T^d)^\nu$, where $(M_H f)(\theta)=H(\theta)f(\theta)$ for $\theta\in\T^d$ and $H(\theta)$ is a $\nu\times\nu$ matrix. Here $\T^d\equiv [0,1)^d$. Denote the eigenvalues of $H(\theta)$ by $E_s(\theta)$, $s=1,\dots \nu$, these are also known as \emph{band functions}, $P_s(\theta)$ the corresponding orthogonal eigenprojections and $P_{E_s}(\theta)$ the orthogonal eigenprojections onto the distinct eigenvalues, that is $P_{E_s}(\theta)=\sum_{w\,:\,E_w(\theta)=E_s(\theta)}P_w(\theta)$. Let $\nu'\le \nu$ be the number of distinct eigenvalues (it is independent of $\theta$ on a subset of full Lebesgue measure \cite[Lem. 2.2]{SY}). Then we have the following: 

\begin{thm}[From \cite{BdMS}]\label{thm:pergra}
Assume that for any $1\le s,w\le \nu$, we have 
\begin{equation}\label{e:flo}
\sup_{\bm\neq 0} \frac{\#\{\br\in \LL_N^d:E_s(\frac{\br+\bm}{N})-E_w(\frac{\mathbf{r}}{N})=0\}}{N^d}\to 0
\end{equation}
as $N\to\infty$. 
%Suppose the observable $a_N$ satisfies one of the following conditions:
%\begin{enumerate}[\rm(i)]
%\item $a_N(k_\fa+v_q) = f^{(q)}(k/N)$ for some $\nu$ functions $f^{(q)}\in H^s(\T_\ast^d)$, with $s>d/2$,
%\item or, $a_N$ is the restriction to $\Gamma_N$ of an integrable function $a\in\ell^1(\Gamma)$.
%\end{enumerate}
Then for a large family of observables,
\[
\lim_{N\to\infty}\Big|\lim_{T\to\infty}\frac{1}{T}\int_0^T\langle \ee^{-\ii tH_N}\delta_{\bn,p}, a\ee^{-\ii tH_N}\delta_{\bn,p}\rangle\,\dd t \\
- \langle a\rangle_p\Big|=0\,,
\]
where, denoting $\langle a_q\rangle:= \frac{1}{N^d} \sum_{\bk\in \LL_N^d} a_q(\bk)$, 
\begin{equation}\label{e:avp}
\langle a\rangle_p = \frac{1}{N^d}\sum_{r\in\LL_N^d} \sum_{q=1}^\nu \langle a_q\rangle \sum_{s=1}^{\nu'} \Big|P_{E_s}\Big(\frac{\br}{N}\Big)(p,q)\Big|^2\,.
\end{equation}
\end{thm}

The initial state $\delta_{\bn,p}\equiv \delta_{v_p}\otimes \delta_{\bn_\fa}$ is a point mass at $v=v_p+\bn_\fa$. That is, $\delta_{\bn,p}$ is the vector on $\Z^d$ will all entries zero except at $\bn$, where the $p$-th coordinate is $1$.\footnote{The theorem is stated slightly differently in \cite[Thm 3.1]{BdMS}, where $\theta$ is written as $\theta_\fb$, in the basis $(\fb_i)$ of coordinates dual to $(\fa_i)$. Here, we use instead the identification \eqref{e:identi} to simplify the notations, as in \cite{SY}. The operator $P(\theta_\fb)$, which was regarded as acting on $\ell^2(V_f)$, is also now regarded as a $\nu\times \nu$ matrix, so that $(P(\theta_\fb)\delta_{v_q})(v_p) = P(\theta)(p,q)$ is the $(p,q)$ entry of the matrix.} Observe that
\[
\langle \ee^{-\ii tH_N}\delta_{\bn,p}, a\ee^{-\ii tH_N}\delta_{\bn,p}\rangle = \sum_{\bk\in\LL_N^d}\sum_{q=1}^\nu a_q(\bk)|(\ee^{-\ii tH_N}\delta_{\bn,p})_q(\bk)|^2
\]
and
\[
\langle a\rangle_p = \frac{1}{N^{2d}}\sum_{\br\in\LL_N^d}\sum_{q=1}^\nu \sum_{\bk\in\LL_N^d} a_q(\bk)\sum_{s=1}^{\nu'}\Big|P_{E_s}\Big(\frac{\br}{N}\Big)(p,q)\Big|^2 \,.
\]
The class of observables is sufficiently large to deduce that if
\[
\mu_{T,v_p+\bn_\fa}^N(\bk_\fa+v_q) :=\frac{1}{T}\int_0^T |(\ee^{-\ii tH_N}\delta_{\bn,p})_q(\bk)|^2\,\dd t = \frac{1}{T}\int_0^T |(\ee^{-\ii tH_N}\delta_{v_p}\mathop\otimes \delta_{\bn_\fa})(\bk_\fa+v_q)|^2\,\dd t\,,
\]
then
\begin{equation}\label{e:genfo}
\mu_{T,v_p+\bn_\fa}^N(\bk_\fa+v_q) \approx \frac{1}{N^{2d}}\sum_{r\in\LL_N^d} \sum_{s=1}^{\nu'}  \Big|P_{E_s}\Big(\frac{\br}{N}\Big)(p,q)\Big|^2 \approx \frac{1}{N^d}\int_{\T^d}\sum_{s=1}^{\nu'}|[P_{E_s}(\theta)(p,q)|^2\,\dd\theta
\end{equation}
as $T\to\infty$ followed by $N\to\infty$, for any $\bn_\fa$ and $\bk_\fa$. This says that the limiting distribution has a density $d(q,p)=\int_{\T^d}\sum_{s=1}^{\nu'}|P_{E_s}(\theta)(p,q)|^2\,\dd\theta$ with respect to the uniform measure. Since all $a_q(\bk)$ are multiplied by the same weight as $\bk$ varies, we see that the weight is fixed on each ``horizontal layer'' $L_q = \{\bk_\fa+v_q:\bk\in \LL_N^d\}$, but the weight can be different from one layer to another, and may depend on the starting point $v_p$, but not on $\bn_\fa$.

Note that $\sum_{q=1}^\nu d(p,q)=\int_{\T^d} \sum_{s=1}^{\nu'} \|P_{E_s}(\theta)\delta_{p}\|^2\,\dd\theta=\int_{\T^d} \|\delta_{p}\|^2\,\dd\theta=1$. So in general, while $\mu_{T,v_p+\bn_\fa}^N(\bk_\fa+v_q)$ may not be perfectly uniform (weight $\frac{1}{|\Gamma_N|}=\frac{1}{\nu N^d}$), the sum $\sum_{q=1}^\nu \mu_{T,v_p+\bn_\fa}^N(\bk_\fa+v_q) \approx \frac{1}{N^d}$ for all $\bn_\fa$, i.e. the total mass in each copy of $V_f$ is constant.

Assumption \eqref{e:flo} essentially says that for any $0\neq \alpha\in\T^d$, the set
\[
\{\theta\in \T^d:E_s(\theta+\alpha)=E_w(\theta)\}
\]
should be an event of measure zero. When applied to $s=w$, this means that $E_s$ should have no nontrivial periods. For $s\neq w$, it is permissible to have $E_s\equiv E_w$, i.e. to have a band function of multiplicity larger than one (as we assume nothing for $\alpha=0$). What is important is that $E_s$ doesn't meet $E_w$ on a \emph{nontrivial} subset of positive measure. This assumption can either be checked by hand if the Floquet matrix is sufficiently simple, or using some tools from the theory of irreducibility of Bloch varieties, see \cite{Wen1} and \cite[\S 5.3]{McKS}.

Assumption \eqref{e:flo} rules out the possibility that $\cH$ has eigenvalues, since $\cH$ has an eigenvalue iff a certain band function $E_s(\theta)\equiv c$ is constant, and if this happens, then the assumption is trivially violated for $w=s$, the underlying fraction being identically $1$ and not converging to zero. In other words, we are already assuming the spectrum of $\cH$ is purely absolutely continuous (AC), which is somehow natural since we are looking for a dynamical delocalization result. However, assumption \eqref{e:flo} is stronger than simply assuming the spectrum is AC, and this stronger assumption seems to be essentially necessary for ergodicity, see \cite[Prp. 1.5]{McKS} for a related result.

Some basic examples were given in \cite[\S 3.2]{BdMS}: if $\nu=1$, i.e. the fundamental cell is simply one vertex, then the limiting measure is uniform. This covers the continuous quantum walk $\ee^{\ii t\cA}$ on $\Z^d$ and the triangular lattice for example. The same property is true for the hexagonal lattice and the infinite ladder, both of which have $\nu=2$. It was also shown that in the cases where $\Gamma$ is a $1d$ strip of width $3$, or a cylinder $\Z\mathop\square C_4$, then the limiting distribution is not uniform.

\subsection{Main results}
In this note we analyze further families of graphs which satisfy our assumptions, and compute the limiting average $\langle a\rangle_p$ explicitly.

We first give the following result, extracted from \cite{McKS}.

\begin{prp}[Case of Cartesian and Tensor Products]\label{prp:mcks}
Suppose $\Gamma_0$ is a periodic graph with $\nu=1$ (for example $\Gamma_0=\Z^d$ or the triangular lattice), and let $G_F$ be any finite graph with $\nu_F=|G_F|$ vertices. Let $\Gamma_1 = \Gamma_0 \mathop\square G_F$ be the Cartesian product , $\Gamma_2 = \Gamma_0 \times G_F$ be the tensor product and $\Gamma_3=\Gamma_0 \boxtimes G_F$ be the strong product of $\Gamma_0$ and $G_F$. Let $E_{\Gamma_0}(\theta)$ be the band function of $\Gamma_0$, $(w_j)$ an orthonormal eigenbasis of $A_{G_F}$ and $(\mu_j)$ the corresponding eigenvalues, $j\le \nu_F$. Then 
\begin{enumerate}
\item The band functions of $\cA_{\Gamma_1}$ are given by $E_j(\theta) = E_{\Gamma_0}(\theta) + \mu_j$.
\item The band functions of $\cA_{\Gamma_2}$ are given by $E_j(\theta) = \mu_j E_{\Gamma_0}(\theta)$.
\item The band functions of $\cA_{\Gamma_3}$ are given by $E_j(\theta) = (1+\mu_j) E_{\Gamma_0}(\theta)+\mu_j$.
\item Assumption~\eqref{e:flo} is satisfied for $\cA_{\Gamma_1}$ but not necessarily for $\cA_{\Gamma_2}$ or $\cA_{\Gamma_3}$.
\item For each of $\cA_{\Gamma_1},\cA_{\Gamma_2}$ and $\cA_{\Gamma_3}$, we have
\[
\langle a\rangle_p = \sum_{q=1}^{\nu_F} \langle a(\cdot+v_q)\rangle \sum_{s=1}^{\nu'_F} |P_{\mu_s}(v_p,v_q)|^2\,,
\]
where $P_{\mu_s}(v_p,v_q) = \sum_{j\,\mu_j=\mu_s} w_j(v_p)\overline{w_j(v_q)}$ is the (kernel) of the orthogonal projection for the distinct eigenvalues of the finite graph.
\end{enumerate}
\end{prp}
For example, if $\Gamma_0=\Z^d$, then $E_{\Gamma_0}(\theta)=2\sum_{i=1}^d \cos 2\pi\theta_i$ and if $\Gamma_0$ is the triangular lattice, then $E_{\Gamma_0}(\theta) = 2\cos 2\pi\theta_1+2\cos 2\pi\theta_2+2\cos 2\pi(\theta_1+\theta_2)$, for $\theta_i\in [0,1)$. Here $\Gamma_{1,2,3}$ are viewed as periodic graphs with fundamental domain containing $\nu_F$ vertices, cf. \cite[Lemma 3.1, \S 3.4]{McKS}, with $(v_i)$ the vertices of $G_F$. 

Arguing as before, we get for these more special graphs that
\begin{equation}\label{e:limcar}
\mu_{T,v_p+\bn_\fa}^N(\bk_\fa+v_q) \approx \frac{1}{N^d} \sum_{s=1}^{\nu'} |P_{\mu_s}(v_p,v_q)|^2 =: \frac{1}{N^d} d(p,q) \,.
\end{equation}
whenever \eqref{e:flo} is satisfied. This gives a more satisfactory concept of a quantum limiting distribution than in \cite{McKS} where quantum ergodicity was assessed by the behaviour of eigenvector bases, and it was shown in \cite[\S 4.5]{McKS} that such a limiting distribution depends on the eigenvector basis. In contrast, here the RHS of \eqref{e:limcar} depends only on the graph.

Our main target now is to compute the weights $d(p,q)$ for specific finite graphs $G_F$. Because of point (4) above, the theorems are illustrated only for the Cartesian product, but some hold more generally. For definiteness, the reader can assume $\Gamma_0=\Z$ in all these results, which is already interesting. However, nothing changes for any $\Gamma_0$ having a single vertex in its fundamental domain, such as $\Z^d$ and the triangular lattice.

A nice simplification in the family of Cartesian products is that the limiting weight $d(p,q)$ in \eqref{e:limcar} depends only on the finite graph $G_F$, compared to the general case \eqref{e:genfo}, where the weight depends on the full Floquet matrix and computations become more daunting. Still, as we will see, Cartesian products already offer interesting contrasting behaviors, depending on the choice of $G_F$. In the following results, $\Gamma_0$ is any $\Z^d$-periodic graph with a single vertex in its fundamental domain.

\subsubsection{Cycles}
We start our analysis with cycles, $G_F=C_\nu$. The graph $\Z\mathop\square C_\nu$ then looks like a cylinder with base $C_\nu$.

\begin{thm}\label{thm:cyc}
Suppose $\Gamma_1= \Gamma_0 \mathop\square C_\nu$, where $C_\nu$ is a cycle of order $\nu$.
\begin{enumerate}
\item If $\nu$ is odd, then $d(p,p) = \frac{2\nu-1}{\nu^2}$ and $d(p,q) = \frac{\nu-1}{\nu^2}$ for all $q\neq p$.
\item If $\nu$ is even, then $d(p,p) = d(p,p+\frac{\nu}{2}) = \frac{2}{\nu}(1-\frac{1}{\nu})$ and $d(p,q) = \frac{1}{\nu}(1-\frac{2}{\nu})$ for all other $q$.
\end{enumerate}
\end{thm}

\subsubsection{Path Graphs}
Consider the path graph $P_\nu$ on $\nu$ vertices numbered $\{1,2,\ldots, \nu\}$.  

\begin{thm}\label{thm:path}
	Suppose $\Gamma_1= \Gamma_0 \mathop\square P_\nu$, where $P_\nu$ is the path graph on $\nu$ vertices. Then
		\[
		d(p,q)=\begin{cases}
			\frac{2}{\nu+1} &\text{ if } p=q=\frac{\nu+1}{2} \\
			\frac{3}{2(\nu+1)}& \text{ if } (p=q,\, p+q \neq \nu+1) \text{ or } (p \neq q,\, p+q=\nu+1)\\
			\frac{1}{\nu+1} &\text{ otherwise. }
		\end{cases}
		\]
\end{thm}
The convention here is that if $\nu$ is even, the first branch never occurs.

\begin{figure}[h!]
	\centering
	\begin{minipage}{.5\textwidth}
		\centering
		\includegraphics[width=0.8\linewidth,height=45mm]{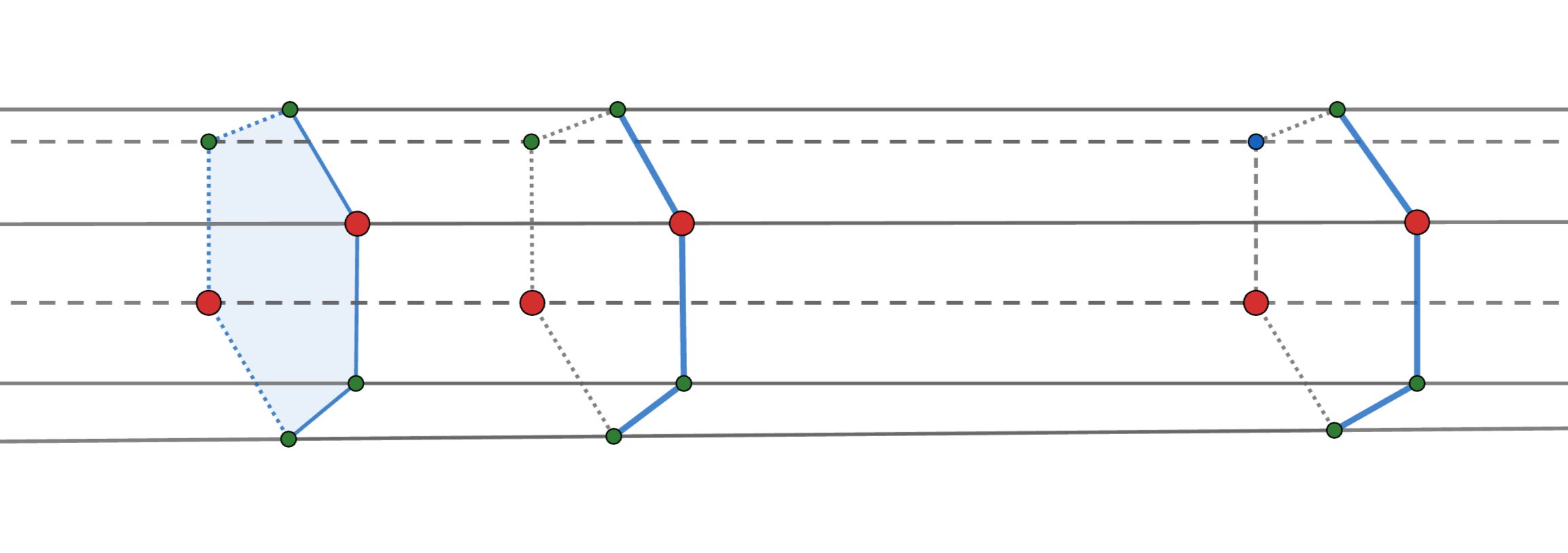}
	\end{minipage}%
	\begin{minipage}{.5\textwidth}
		\centering
		\includegraphics[width=0.8\linewidth,height=45mm]{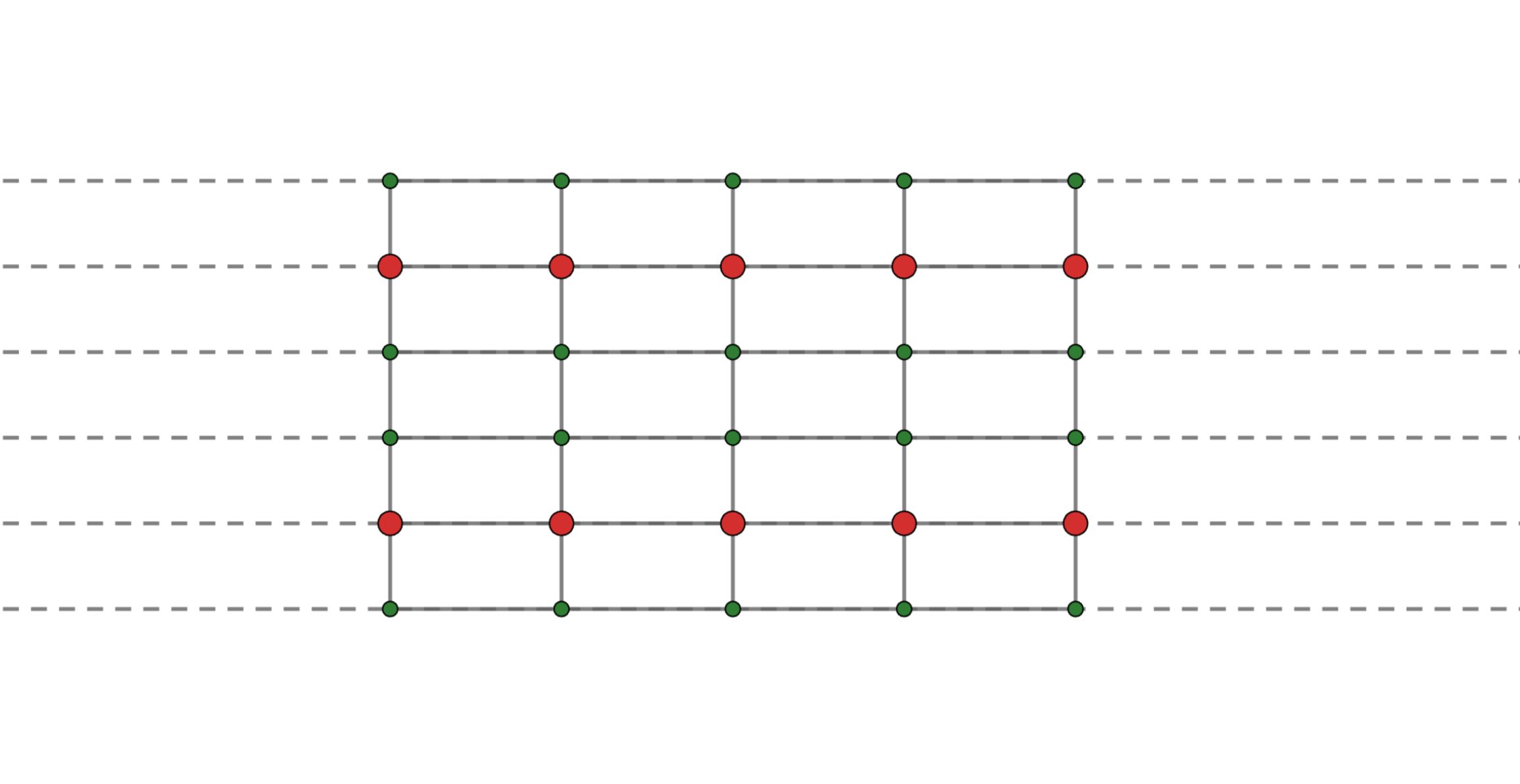}
	\end{minipage}
	\caption{The limiting distribution of the continuous-time quantum walk on graphs of the form $\Z \mathop\square  G_F$ for $G_F$ as the cycle $C_6$ (left) and the path $P_6$ (right). In both figures, the walk starts at a red vertex, and it is the vertex where most of the weight concentrates. All the green vertices have the same weight, lower than the red vertices.}
	\label{fig:path&cycle}
\end{figure}

\subsubsection{Stars}
Star graphs behave very differently, see \S\,\ref{sec:unders} for more comments.
\begin{thm}\label{thm:star}
		Suppose $\Gamma_1= \Gamma_0 \mathop\square K_{\nu,1}$, where $K_{\nu,1}$ is the star graph on $\nu+1$ vertices with $\nu+1$ as the central vertex. Then
		\[
		d(p,q) = \begin{cases} \frac{(\nu-1)^2}{\nu^2}+\frac{1}{2\nu^2}& \text{if } 1\le p\le \nu \text{ and } q=p,\\ \frac{1}{2\nu}& \text{if } 1\le p\le \nu \text{ and } q=\nu+1,\\ \frac{3}{2\nu^2} & \text{if } 1\le p\le \nu \text{ and } 1\le q \le \nu \text{ and } q\neq p,\\ \frac{1}{2}& \text{ if } p=q=\nu+1. \end{cases}
		\]
\end{thm}

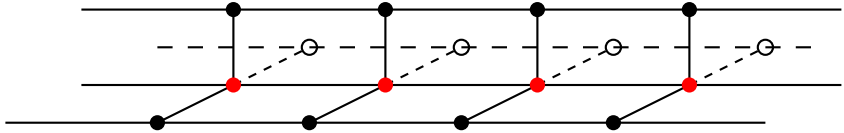
\begin{figure}[h!]
\begin{center}
\setlength{\unitlength}{1cm}
\thicklines
\begin{picture}(1.3,1.5)(-1.3,-1.5)
   \put(-5,0){\line(1,0){10}}
	 \put(-5,-1){\line(1,0){10}}
	 \put(-3,-1){\line(0,1){1}}
	 \put(-1,-1){\line(0,1){1}}
	 \put(1,-1){\line(0,1){1}}
	 \put(3,-1){\line(0,1){1}}
	 \put(1,-1){\line(-1,-0.5){1}}
	 \put(3,-1){\line(-1,-0.5){1}}
 	 \put(-1,-1){\line(-1,-0.5){1}}
	 \put(-3,-1){\line(-1,-0.5){1}}
	\multiput(-1,-1)(0.2,0.1){5}{\line(1,0.5){0.1}}
	\multiput(1,-1)(0.2,0.1){5}{\line(1,0.5){0.1}}
	\multiput(3,-1)(0.2,0.1){5}{\line(1,0.5){0.1}}
	\multiput(-3,-1)(0.2,0.1){5}{\line(1,0.5){0.1}}
	 \put(-1,-1){\textcolor{red}{\circle*{.2}}}
	 \put(-1,0){\circle*{.2}}
	 \put(-3,-1){\textcolor{red}{\circle*{.2}}}
	 \put(-3,0){\circle*{.2}}
	 \put(1,-1){\textcolor{red}{\circle*{.2}}}
	 \put(1,0){\circle*{.2}}
	 \put(3,-1){\textcolor{red}{\circle*{.2}}}
	 \put(3,0){\circle*{.2}}
	\put(-6,-1.5){\line(1,0){10}}
	\multiput(-4,-0.5)(0.4,0){22}{\line(1,0){0.2}}
	\put(-2,-1.5){\circle*{.2}}
	\put(-4,-1.5){\circle*{.2}}
	\put(0,-1.5){\circle*{.2}}
	\put(2,-1.5){\circle*{.2}}
	\put(-2,-0.5){\circle{.2}}
	\put(0,-0.5){\circle{.2}}
	\put(2,-0.5){\circle{.2}}
	\put(4,-0.5){\circle{.2}}
\end{picture}
\caption{The graph $\Z\mathop\square K_{3,1}$. No matter how many edges the star $K_{\nu,1}$ has, if we launch the walk from a red vertex, half the mass spreads over the line of red vertices, while the remainder gets equidistributed over the leaves. The situation is more dramatic if the walk is launched from a leaf vertex: if $\nu=10$, then $0.815$ of the mass spreads over the layer of this leaf.}\label{fig:star}
\end{center}
\end{figure}

\subsubsection{Hypercubes}
Recall that the $m$-dimensional hypercube $\mathbf{H}_m$ has vertex set $\Z_2^m$. Two vertices are adjacent if they differ by a single digit, hence $\mathbf{H}_m$ is $m$-regular. Graph products $\Gamma_0\mathop\square \mathbf{H}_m$ have $G_F = \mathbf{H}_m$ as fundamental domain. Up to renumbering the vertices within $G_F$, we may start the quantum walk from $v_p = (0,0,\dots,0)$.

\begin{thm}\label{thm:hyp}
Suppose $\Gamma_1 = \Gamma_0  \mathop\square \mathbf{H}_m$, where $\mathbf{H}_m = \Z_2^m$ is the $m$-dimensional hypercube. Let $v_0=(0,\dots,0)$. Define $B_u\subset \Z_2^m$ as the subset of vertices having exactly $u$ entries equal to $1$. Then for any $v_q\in B_u$,
\begin{equation}\label{e:weighhyp}
d(0,q) = \frac{1}{2^{2m}} \sum_{j=0}^m \left(\sum_{b=0}^u (-1)^{b}\binom{u}{b}\binom{m-u}{j-b}\right)^2\,.
\end{equation}

The convention here is that $\binom{m-u}{j-b}=0$ if $j<b$ or $j-b>m-u$.

In particular, $d(0,0)=\frac{1}{2^{2m}}\binom{2m}{m} = d(0,m)$, while $d(0,q) = \frac{1}{2^{2m}}[\binom{2m}{m}-4\binom{2m-1}{m-1}+4\binom{2m-2}{m-1}]$ for $v_q\in B_1$. We also have the symmetry $d(0,q)=d(0,q')$ if $v_q\in B_u$ and $v_{q'}\in B_{m-u}$.
\end{thm}

Let us calculate \eqref{e:weighhyp} explicitly in simple cases. The graphs $\mathbf{H}_1$ and $\mathbf{H}_2$ give a segment $P_2$ and a $4$-cycle $C_4$, respectively, which are covered by previous theorems. So let us consider the cases of the cube and tesseract:

\begin{cor}
Define $B_u$ as in Theorem~\ref{thm:hyp}. In case of the cube $\mathbf{H}_3$, we get
\[
d(0,q)=\begin{cases}
	5/16 &\text{ if } q \in B_0 \\
	1/16 &\text{ if } q \in B_1 \\
	1/16 &\text{ if } q \in B_2 \text{ and}\\
	5/16 &\text{ if } q \in B_3.
\end{cases}
\]

In case of the tesseract $\mathbf{H}_4$, we get
\[
d(0,q)=\begin{cases}
	35/128 &\text{ if } q \in B_0 \\
	5/128 &\text{ if } q \in B_1 \\
	3/128 &\text{ if } q \in B_2 \\
	5/128 &\text{ if } q \in B_3 \text{ and}\\
	35/128 &\text{ if } q \in B_4.
\end{cases}
\]
\end{cor}

\subsubsection{Understanding the results}\label{sec:unders}
The previous theorems illustrate that the limiting measure, though spread out (all vertices carry a weight $\asymp N^{-d}$), is not uniform: the weight $d(p,q)\neq \frac{1}{\nu}$, instead, it depends on the initial point $p$, and the masses vary over the fundamental domain (as $q$ varies).

However, there is a large contrast in the limiting behaviors across the graphs: in case of $\Gamma_0\mathop\square C_\nu$ and $\Gamma_0\mathop \square P_\nu$, the mass is almost uniform, $d(p,q)$ is of order $\frac{1}{\nu}$ everywhere. In case of $\Gamma_0\mathop\square \mathbf{H}_m$ on the other hand, the mass is highly concentrated on the starting vertex $\mathbf{0}=(0,\dots,0)$ and the one diametrically opposite to it, $\mathbf{1}=(1,\dots,1)$: on both vertices, the mass if $\frac{1}{2^{2m}}\binom{2m}{m}$, which is of order $\frac{1}{\sqrt{m}}$, much larger than the uniform weight $\frac{1}{2^m}$. The situation is even worse in star graphs $K_{1,\nu}$, where if we take $\nu\gg 1$, the mass concentrates macroscopically over the line from which the walk was launched.

We can say that in Theorems \ref{thm:cyc} andd \ref{thm:path}, we have ergodicity both in the ``horizontal'' direction (the mass spreads as $N^{-d}$) and the ``vertical'' or ``sectional'' direction (all vertices actually carry a weight $\approx \frac{1}{\nu N^d}$), while in Theorems \ref{thm:star} and \ref{thm:hyp}, ergodicity only holds in the horizontal direction.

Theorems \ref{thm:cyc} and \ref{thm:hyp} offer a stark contrast with the following:

\medskip

\textbf{Fact.} Consider a classical random walk on a finite graph with $n$ vertices and $m$ edges, where a particle at vertex $v$ hops uniformly at random to a neighbor $w$ with probability $\frac{1}{\deg(v)}$. Then the distribution $\pi(v)=\frac{\deg(v)}{2m}$ is stationary, and if the graph is non-bipartite, this distribution is also the limiting distribution that the particle ends up at $v$, i.e. $\lim_{n\to \infty} p_{u,v}(n) = \pi(v)$.

\medskip

See \cite{Lov} for details. In particular, if the graph is $d$-regular, the limiting distribution of the classical walk must be uniform if it exists, $\pi(v) = \frac{d}{2m}=\frac{1}{n}$. In our case, the graphs $G_F = C_\nu$ and $G_F=\mathbf{H}_m$ are regular. Since we are considering periodic boundary conditions, this implies the corresponding graphs $\Gamma_N$ are also regular. So contrary to Theorems \ref{thm:cyc} and \ref{thm:hyp}, the classical walk would equidistribute on these graphs. The theorems seem to indicate an interesting quantum effect.

\section{Proofs}
%We begin by briefly pointing out the proof of Proposition~\ref{prp:mcks}.

\begin{proof}[Proof of Proposition~\ref{prp:mcks}]
Claims (1) and (2) are proved in \cite{McKS}, Sections 3.2 and 3.4, respectively and (3) follows directly from (1) and (2). For point (4), the validity of the Floquet assumption for the Cartesian product is part of Proposition 1.3 in \cite{McKS}, while its failure is obvious in the tensor product case if $G_F$ has an eigenvalue $\mu_j=0$ (e.g. $G_F=C_4$, the $4$-cycle), since in this case $E_j(\theta) = E_{\Gamma_0}\mu_j=0$ is constant, so that $\Gamma_2$ has a flat band and the Floquet assumption is violated. Similarly, if $\mu_j=-1$, we get a flat band for $\Gamma_3$.

To see the validity of Claim (5), note that the Floquet matrices for $\Gamma_1$ and $\Gamma_2$ are given by $A_{\Gamma_1}(\theta) = E_{\Gamma_0}(\theta)I_{\nu_F} + A_{G_F}$ and $A_{\Gamma_2}(\theta) = E_{\Gamma_0}(\theta)A_{G_F}$, respectively. In particular, they share the same eigenvectors as $A_{G_F}$, with $A_{\Gamma_1}(\theta)w_j = (E_{\Gamma_0}+\mu_j)w_j$ and $A_{\Gamma_2}(\theta)w_j = E_{\Gamma_0}(\theta)\mu_j w_j$. This gives $P_s(\theta) = \langle \cdot, w_s\rangle w_s$. Moreover, $E_j(\theta)=E_i(\theta)$ iff $\mu_j = \mu_i$. So $P_{E_s}(\theta) = \sum_{w\,:\,\mu_s=\mu_j} \langle \cdot, w_j\rangle w_j$ and $[P_{E_s}(\theta)\delta_{v_q}](v_p) = \sum_{w\,:\,\mu_s=\mu_j} w_j(v_p)\overline{w_j(v_q)}$.
\end{proof}

%\subsection{Case of Cayley Graphs}
%Let $\cG$ be a group and $S\subseteq \cG$. We denote the group operation by addition. The Cayley graph $\cC$ of $(\cG,S)$ is the graph with vertices $\cG$ and edges $\{x,x+g\}$, where $g\in S$. We assume $S=-S$ so that the graph is undirected.
%
%For example, if $\cG = \Z_\nu$ and $S = \{\pm 1\}$, then $\cC$ is the $\nu$-cycle.
%
%If $\cG = \Z_n\times \Z_m$ and $S = \{(0,\pm 1), (\pm 1,0)\}$, then $\cC$ is the $n\times m$ grid graph on the $2d$ torus.
%
%If $\cG = \Z_2^m$ and $S = \{e_1,\dots,e_n\}$, where $e_i(k) = \delta_{i,k}$, then $\cC$ is the $m$-dimensional hypercube.
%
%A \emph{character on $\cG$} is a map $\chi:\cG\to \C^\ast$ satisfying $\chi(x+y) = \chi(x)\chi(y)$. If is known that a finite group $\cG$ has $n=|\cG|$ distinct characters. Moreover, for an abelian group $\cG$ and $S \subseteq \cG$, the vector $w_\chi=(\chi(x))_{x\in \cG}$ is an eigenvector of $\cA_{\cC(\cG,S)}$ with eigenvalue $\sum_{s\in S} \chi(S)$ [\cite{Cayley_survey},Corollary 3].

%If $\cG=\Z_n$, then the distinct characters take the form $\chi_r(j) = \ee^{\frac{2\pi\ii r j}{n}}$, for $r,j=0,1,\dots,n-1$. Assuming $S= \{\pm 1\}$, the corresponding eigenvalues are thus $\mu_r = 2\cos(\frac{2\pi r}{n})$, $r=0,\dots,n-1$ \cite{Cayley_survey}.

\begin{proof}[Proof of Theorem~\ref{thm:cyc}]
As the cycle is homogeneous, we may assume without loss that the starting point $p=0$ (enumerate the vertices starting from $v_p$).

The eigenvalues of $C_\nu$ are $\mu_r = 2\cos(\frac{2\pi r}{\nu})$, with corresponding orthonormal eigenvectors $w_r = \frac{1}{\sqrt{\nu}}(1,\ee^{\frac{2\pi\ii r}{\nu}}, \ee^{\frac{2\pi\ii (2r)}{\nu}},\dots, \ee^{\frac{2\pi \ii (\nu-1)r}{\nu}})$.

Note that $\mu_r = \mu_{\nu-r}$ for $1\le r\le \nu-1$.

$\bullet$~If $\nu$ is odd, then $\mu_0=2$ is simple and each $\mu_r$ with $1\le r\le \frac{\nu-1}{2}$ has multiplicity $2$. The eigenvector for $\mu_0$ is $w_0 = \frac{1}{\sqrt{\nu}}(1,\dots,1)$, and for $\mu_r$, we have the two eigenvectors $w_r = \frac{1}{\sqrt{\nu}}(1,\ee^{\frac{2\pi\ii r}{\nu}}, \ee^{\frac{2\pi\ii (2r)}{\nu}},\dots, \ee^{\frac{2\pi \ii (\nu-1)r}{\nu}})$ and $w_{\nu-r} = \frac{1}{\sqrt{\nu}}(1,\ee^{\frac{-2\pi\ii r}{\nu}}, \ee^{\frac{-2\pi\ii (2r)}{\nu}}, \dots, \ee^{-2\pi\ii\frac{(\nu-1)r}{\nu}})$. We deduce that $P_{\mu_0}(v_p,v_q) = w_0(v_p)\overline{w_0(v_q)}$ and $P_{\mu_r}(v_p,v_q) = w_r(v_p)\overline{w_r(v_q)}+w_{\nu-r}(v_p)\overline{w_{\nu-r}(v_q)}$ for $1\le r\le \frac{\nu-1}{2}$.

Taking $p=0$, we get
\begin{align*}
\sum_{s=0}^{\frac{\nu-1}{2}} |P_{\mu_s}(v_0,v_q)|^2 &= \frac{1}{\nu^2} + \frac{1}{\nu^2} \sum_{r=1}^{\frac{\nu-1}{2}} |\overline{w_r(v_q)+w_{\nu-r}(v_q)}|^2 = \frac{1}{\nu^2} + \frac{4}{\nu^2}\sum_{r=1}^{\frac{\nu-1}{2}} \cos^2\Big(\frac{2\pi rq}{\nu}\Big) \\
&= \frac{1}{\nu^2}+\frac{2}{\nu^2} \sum_{r=1}^{\frac{\nu-1}{2}} \Big[\cos\Big(\frac{4\pi rq}{\nu}\Big)+1\Big] = \frac{1}{\nu} + \frac{2}{\nu^2} \sum_{r=1}^{\frac{\nu-1}{2}} \cos\Big(\frac{4\pi rq}{\nu}\Big)\,.
\end{align*}

Now recall the identity $\sum_{k=0}^n \cos k\theta = \frac{\sin(n+\frac{1}{2})\theta+\sin \frac{\theta}{2}}{2\sin\frac{\theta}{2}}$ for $\theta\notin 2\pi\Z$. We deduce that $\sum_{r=1}^{\frac{\nu-1}{2}} \cos (r\theta) = \frac{\sin \frac{\nu\theta}{2}-\sin\frac{\theta}{2}}{2\sin\frac{\theta}{2}}$. For $\theta=\frac{4\pi q}{\nu}$, $q\neq 0$, this gives $\frac{-1}{2}$. Thus,
\[
\sum_{s=0}^{\frac{\nu-1}{2}} |P_{\mu_s}(v_0,v_q)|^2 = \begin{cases} \frac{1}{\nu} + \frac{\nu-1}{\nu^2}&\text{if }q=0,\\ \frac{1}{\nu}-\frac{1}{\nu^2}&\text{if }q\neq 0\end{cases} = \begin{cases} \frac{2\nu-1}{\nu^2} &\text{if }q=0,\\ \frac{\nu-1}{\nu^2}&\text{if }q\neq 0.\end{cases}
\]

$\bullet$~Now suppose $\nu$ is even. Then $\mu_0=2$ and $\mu_{\frac{\nu}{2}} = -2$ are simple, while $\mu_r = 2\cos(\frac{2\pi r}{\nu})$ has multiplicity $2$ for $1\le r<\frac{\nu}{2}$. It follows that
\[
\sum_{s=0}^{\frac{\nu}{2}} |P_{\mu_s}(v_0,v_q)|^2 = \frac{2}{\nu^2} + \frac{4}{\nu^2}\sum_{r=1}^{\frac{\nu}{2}-1} \cos^2\Big(\frac{2\pi rq}{\nu}\Big) = \frac{1}{\nu} + \frac{2}{\nu^2}\sum_{r=1}^{\frac{\nu}{2}-1}\cos\Big(\frac{4\pi r q}{\nu}\Big) \,.
\]
As before, $\sum_{r=1}^{\frac{\nu}{2}-1}\cos(r\theta) = \frac{\sin\frac{\nu-1}{2}\theta - \sin\frac{\theta}{2}}{2\sin\frac{\theta}{2}}$, yielding $-1$ if $\theta=\frac{4\pi q}{\nu}$. Thus,
\[
\sum_{s=0}^{\frac{\nu}{2}} |P_{\mu_s}(v_0,v_q)|^2 = \begin{cases} \frac{1}{\nu} + \frac{2}{\nu^2}(\frac{\nu}{2}-1)&\text{if }q=0,\frac{\nu}{2}\\ \frac{1}{\nu}-\frac{2}{\nu^2}&\text{if }q\neq 0, \frac{\nu}{2}\end{cases} = \begin{cases} \frac{2}{\nu}(1-\frac{1}{\nu}) &\text{if }q=0,\frac{\nu}{2}\\ \frac{1}{\nu}(1-\frac{2}{\nu})&\text{if }q\neq 0,\frac{\nu}{2}.\end{cases} \qedhere
\]
\end{proof}

\begin{proof}[Proof of Theorem~\ref{thm:path}]
We have 
$\sigma(P_\nu)=\{2\cos\big(\frac{\pi j}{\nu+1}\big): 1 \leq j \leq \nu\},$
with the eigenvector corresponding to $2\cos(\frac{\pi j}{\nu+1})$ denoted by $w_j=\sqrt{\frac{2}{\nu+1}}\big(\sin(\frac{\pi j}{\nu+1}),\sin(\frac{2 \pi j}{\nu+1}),\ldots,\sin(\frac{\nu \pi j }{\nu+1})\big).$

Each eigenvalue of $P_\nu$ has multiplicity 1, as $\frac{\pi j}{\nu+1}$ takes values in the first two quadrants, so the $w_j$ are orthogonal, and they are normalized to have $\|w_j\|=1$. Indeed, $\sum_{\ell=1}^\nu \sin^2(\frac{\pi j \ell}{\nu+1})=\frac{\nu+1}{2}$, as we see by following the same calculations as in the proof of Theorem~\ref{thm:cyc}. Therefore, $P_{\mu_j}(p,q)=w_j(p)w_j(s)$. Expanding $\sin^2 x = \frac{1-\cos 2x}{2}$, we get
\begin{equation}\label{e:sumeig}
\sum_{s=1}^\nu |P_{\mu_s}(p,q)|^2 =\frac{4}{(\nu+1)^2}\sum_{s=1}^\nu \Big(\frac{1-\cos(\frac{2\pi sq}{\nu+1})}{2}\Big)\Big(\frac{1-\cos(\frac{2\pi sp}{\nu+1})}{2}\Big) \,.
%= \sum_{s=1}^\nu \frac{4}{(\nu+1)^2}\sin^2\Big(\frac{\pi s q}{\nu+1}\Big)\sin^2\Big(\frac{\pi s p}{\nu+1}\Big)
\end{equation}
Using $\sum_{s=1}^\nu \cos s\theta = \frac{\sin(\nu+\frac{1}{2})\theta - \sin\frac{\theta}{2}}{2\sin\frac{\theta}{2}}$, we get $\sum_{s=1}^\nu \cos s(\frac{2\pi \ell}{\nu+1}) = \frac{\sin(2\pi \ell-\frac{\pi \ell}{\nu+1})-\sin\frac{\pi\ell}{\nu+1}}{2\sin\frac{\pi\ell}{\nu+1}} = -1$, for $\ell=p,q$. Next, $\cos(\frac{2\pi sq}{\nu+1})\cos(\frac{2\pi sp}{\nu+1})=\frac{\cos\frac{2\pi s}{\nu+1}(p+q) + \cos\frac{2\pi s}{\nu+1}(p-q)}{2}$. Therefore,
\[
\sum_{s=1}^\nu \cos\Big(\frac{2\pi sq}{\nu+1}\Big)\cos\Big(\frac{2\pi sp}{\nu+1}\Big) = \begin{cases} \nu& \text{if } p+q=\nu+1 \text{ and } p-q=0,\\ \frac{\nu-1}{2}&\text{if } (p+q=\nu+1,\,p\neq q) \text{ or } (p+q\neq \nu+1,\,p=q),\\ -1&\text{ otherwise.}\end{cases}
\]
Recalling \eqref{e:sumeig}, we have shown that
\[
\sum_{s=1}^\nu |P_{\mu_s}(p,q)|^2 =\frac{1}{(\nu+1)^2}\Big(\nu+1+1+\sum_{s=1}^\nu\cos\Big(\frac{2\pi sq}{\nu+1}\Big)\cos\Big(\frac{2\pi sp}{\nu+1}\Big)\Big)\,.
\]
Collecting the estimates completes the proof.
\end{proof}

\begin{proof}[Proof of Theorem~\ref{thm:star}]
In the adjacency matrix $A_{K_{\nu,1}}$, the upper $\nu \times \nu$ submatrix is a zero matrix, and the last row is all ones except the last element. It follows that $w_r = \frac{1}{\sqrt{\nu}}(1,\ee^{\frac{2\pi\ii r}{\nu}}, \ee^{\frac{2\pi\ii (2r)}{\nu}},\dots, \ee^{\frac{2\pi \ii (\nu-1)r}{\nu}},0)$ is an eigenvector with eigenvalue 0 for $1 \leq r \leq \nu-1$. The other eigenvalues are $\lambda_{\pm}=\pm\sqrt{\nu}$ with eigenvectors $w_{\pm}=(\frac{\pm1}{\sqrt{2\nu}},\frac{\pm 1}{\sqrt{2\nu}},\ldots,\frac{\pm1}{\sqrt{2\nu}},\frac{1}{\sqrt{2}})$. Thus,
\[
d(p,q) = |P_0(p,q)|^2 + |P_{\sqrt{\nu}}(p,q)|^2 + |P_{-\sqrt{\nu}}(p,q)|^2 \,.
\]
We have $P_0(p,q) = \sum_{r=1}^{\nu-1} w_r(v_p)\overline{w_r(v_q)}$.

Suppose $1\le p\le \nu$. Then $P_0(p,p)=\frac{\nu-1}{\nu}$ and $P_{\sqrt{\nu}}(p,p)=P_{-\sqrt{\nu}}(p,p)=\frac{1}{2\nu}$. Thus, $d(p,p) = \frac{(\nu-1)^2}{\nu^2}+\frac{1}{2\nu^2}$. Next, $P_0(p,\nu+1)=0$, $P_{\sqrt{\nu}}(p,\nu+1)=\frac{1}{2\sqrt{\nu}}$ and $P_{-\sqrt{\nu}}(p,\nu+1)=\frac{-1}{2\sqrt{\nu}}$. Thus, $d(p,\nu+1) = \frac{1}{2\nu}$. Finally, by symmetry all vertices $q\notin \{p,\nu+1\}$ must have the same weight $M$. Since $\sum_{q=1}^{\nu+1} d(p,q)=1$, this leads to $\frac{(\nu-1)^2}{\nu^2}+\frac{1}{2\nu^2}+\frac{1}{2\nu} + (\nu-1)M=1$. This implies $M=\frac{3}{2\nu^2}$.

Now suppose $p=\nu+1$. Then $d(\nu+1,p)=d(p,\nu+1)=\frac{1}{2\nu}$ for any $1\le p\le \nu$, since the weight function $d$ is symmetric in its arguments. This implies $d(\nu+1,\nu+1)=\frac{1}{2}$ using $\sum_{q=1}^{\nu+1} d(p,q)=1$, or by direct calculation.
\end{proof}

\begin{proof}[Proof of Theorem~\ref{thm:hyp}]
The hypercube $\mathbf{H}_m$ is the Cayley graph of the group $\cG = \Z_2^m$, with generators $S = \{e_1,\dots,e_m\}$, where $e_i(k) = \delta_{i,k}$. As such, the vector of characters $w_\chi = (\chi(x))_{x\in \cG}$ is an eigenvector, with eigenvalue $\sum_{s\in S} \chi(s)$. The characters here are given by $\chi_r(x) = \ee^{\frac{2\pi\ii r\cdot x}{2}} = (-1)^{r_1x_1+\dots+r_mx_m}$, see \cite[\S 2.2]{Cayley_survey}. In particular,
\[
\sum_{s\in S} \chi_r(s) = \sum_{i=1}^m (-1)^{r_i} = \#\{ r_i=0\}-\#\{r_i=1\} = (m-k)-k = m-2k\,,
\]
where $k=\#\{r_i=1\}$. We thus get $\sigma(\cA_{\cC}) = \{-m,-m+2,\dots,0,2,\dots,m\}$ for even $m$ and $\sigma(\cA_{\cC}) = \{-m,-m+2,\dots,-1,1,\dots,m\}$ for odd $n$. The multiplicity of $\mu_k:= m-2k$ is $\#\{r \text{ having exactly } k \text{ entries } 1\} = \binom{m}{k}$.

Recall that $w_\chi = (\chi(x))$ are the eigenvectors. The eigenvalues $\mu_0$ and $\mu_m$ are simple, with $P_{\mu_0}(v_0,v_q) = w_0(v_q)w_0(v_0) = \frac{(-1)^0}{2^m}=\frac{1}{2^m}$ and $P_{\mu_m}(v_0,v_q) = w_m(v_q)w_m(v_0)=\frac{(-1)^{q_1+\dots+q_m}}{2^m}$. For the remaining eigenvalues $\mu_k$, we have
\begin{equation}\label{e:hypercomp}
P_{\mu_k}(v_0,v_q)= \sum_{r \in B_k}w_r(v_q)w_r(v_0)=\sum_{r \in B_k}\frac{1}{2^m}(-1)^{r_1q_1+\cdots+r_mq_m} \,.
\end{equation}

There are two easy cases: if $v_q\in B_0$, i.e. $v_q=(0,\dots,0)$, we get $\frac{1}{2^m} \binom{m}{k}$. Similarly, if $v_q\in B_m$, i.e. $v_q=(1,\dots,1)$, we get $\frac{1}{2^m}\sum_{r\in B_k}(-1)^{r_1+\dots+r_m} = \frac{(-1)^k}{2^m}\binom{m}{k}$.

To compute \eqref{e:hypercomp} for $v_q\in B_u$ in general, we first note that since the sum is taken over all $r\in B_k$, the expression is invariant under permutations of the vertices of $v_q$, so we may choose $v_q=(1,\dots,1,0,\dots,0)$ to simplify the notation, where $v_q$ contains $u$ entries $1$. Expression \eqref{e:hypercomp} becomes
\[
P_{\mu_k}(v_0,v_q)= \frac{1}{2^m}\sum_{r\in B_k} (-1)^{r_1+\dots+r_u} \,.
\]

Let us introduce, for $0\le b\le u$, the set
\[
F_{u,k}(b) = \{r\in B_k: r_1+\dots+r_u=b\} \,.
\]

Clearly, $\{F_{u,k}(b)\}_{b=0}^u$ is a partition of $B_k$. Hence,
\[
P_{\mu_k}(v_0,v_q)= \frac{1}{2^m} \sum_{b=0}^u \sum_{r\in F_{u,k}(b)} (-1)^{r_1+\dots+r_u} = \frac{1}{2^m}\sum_{b=0}^u (-1)^b \# F_{u,k}(b) \,.
\]

Now note that
\begin{enumerate}[(i)]
\item The number of ones in $\{r_1,\dots,r_u\}$ is $b$.
\item The number of ones in $\{r_1,\dots,r_m\}$ is $k$.
\item Hence, the number of ones in $\{r_{u+1},\dots,r_m\}$ is $k-b$.
\end{enumerate}

We now have several cases:
\begin{itemize}
\item If $k<b$, there are more ones in (i) than in (ii), which is impossible. Thus, $F_{u,k}(b)$ is empty and $\# F_{u,k}(b) = 0 = \binom{u}{b}\binom{m-u}{k-b}$, since $k-b$ is negative.
\item If $k-b>m-u$, we have more ones in (iii) than the cardinality of the set, which is impossible. So again, we get $\# F_{u,k}(b) = 0 = \binom{u}{b}\binom{m-u}{k-b}$.
\item If $0\le k-b\le m-u$, with $0\le b \le u$, then each $r\in F_{u,k}(b)$ is given by a choice of $b$ ones from $\{r_1,\dots,r_u\}$ and a choice of $k-b$ ones from $\{r_{u+1},\dots,r_m\}$. Therefore, $\# F_{u,k}(b) = \binom{u}{b}\binom{m-u}{k-b}$.
\end{itemize}

We have shown that in all cases,
\[
P_{\mu_k}(v_0,v_q) = \frac{1}{2^m} \sum_{b=0}^u (-1)^b \binom{u}{b}\binom{m-u}{k-b}\,.
\]
Since $d(0,q) = \sum_{k=0}^m |P_{\mu_k}(v_0,v_q)|^2$, this completes the proof of the general form.

The particular cases follow from the identity $\sum_{k=0}^m \binom{m}{k}^2=\binom{2m}{m}$, which settles $d(0,0)$ and $d(0,m)$. For $v_q\in B_1$, we get $\frac{1}{2^{2m}}\sum_{j=0}^m (\binom{m-1}{j}-\binom{m-1}{j-1})^2 = \frac{1}{2^{2m}}\sum_{k=0}^m(\binom{m}{j}(1-\frac{2j}{m}))^2 = \frac{1}{2^{2m}}\sum_{j=0}^m [\binom{m}{j}^2 - \frac{4j}{m}\binom{m}{j}^2 + \frac{4j^2}{m^2}\binom{m}{j}^2] = \frac{1}{2^{2m}} [\binom{2m}{m}-4\binom{2m-1}{m-1}+4\binom{2m-2}{m-1}]$.

The symmetry $B_u\leftrightarrow B_{m-u}$ may be seen directly from \eqref{e:hypercomp}. If $v_q\in B_u$, then \eqref{e:hypercomp} takes the form $\frac{1}{2^m}\sum_{r\in B_k} (-1)^{r_{i_1}+\dots+r_{i_u}}$. If we consider $v_{q'} = (1,\dots,1)-v_q\in B_{m-u}$, we get $\frac{1}{2^m}\sum_{r\in B_k} (-1)^{r_{i_u+1}+\dots+r_{i_m}}$. Here $(i_k)_{k=1}^u$ enumerate the ones of $v_q$ and $(i_k)_{k=u+1}^m$ enumerate the the zeros of $v_q$ (which are the ones of $v_{q'}$). However, $(-1)^{r_{i_u+1}+\dots+r_{i_n}} = (-1)^{2r_{i_1}+\dots+2r_{i_u}}(-1)^{r_{i_u+1}+\dots+r_{i_n}} = (-1)^{r_1+\dots+r_n}(-1)^{r_{i_1}+\dots+r_{i_u}} = (-1)^k (-1)^{r_{i_1}+\dots+r_{i_u}}$. Thus, $P_{\mu_k}(v_0,v_{q'}) = (-1)^k P_{\mu_k}(v_0,v_q)$, implying $d(0,q)=d(0,q')$.
\end{proof}

\section{Further examples}

\subsubsection{Complete graphs and complete bipartite graphs}
For complete graphs, and complete bipartite graphs, we observe most of the mass concentrating on the starting vertex. For the complete graph $K_\nu$, we denote the function $d: [\nu] \times [\nu] \rightarrow [0,1]$ defined in \eqref{e:limcar} by $d_\nu$. Due to the symmetry of the graph, $d_\nu(p,p)=d_\nu(q,q)$ for all $p,q$ and $d_{\nu}(p_1,q_1)=d_\nu(p_2,q_2)$ for all $p_1 \neq q_1,p_2 \neq q_2$. By direct computation we get $d_4(p,p)=5/8,d_5(p,p)=0.68,d_6(p,p) \approx 0.722$, $d_7(p,p)\approx 0.75$ and $d_{100}(p,p) \approx 0.98$.

\begin{figure}[h!]
	\centering
	\begin{minipage}{.5\textwidth}
		\centering
		\includegraphics[width=0.75\linewidth]{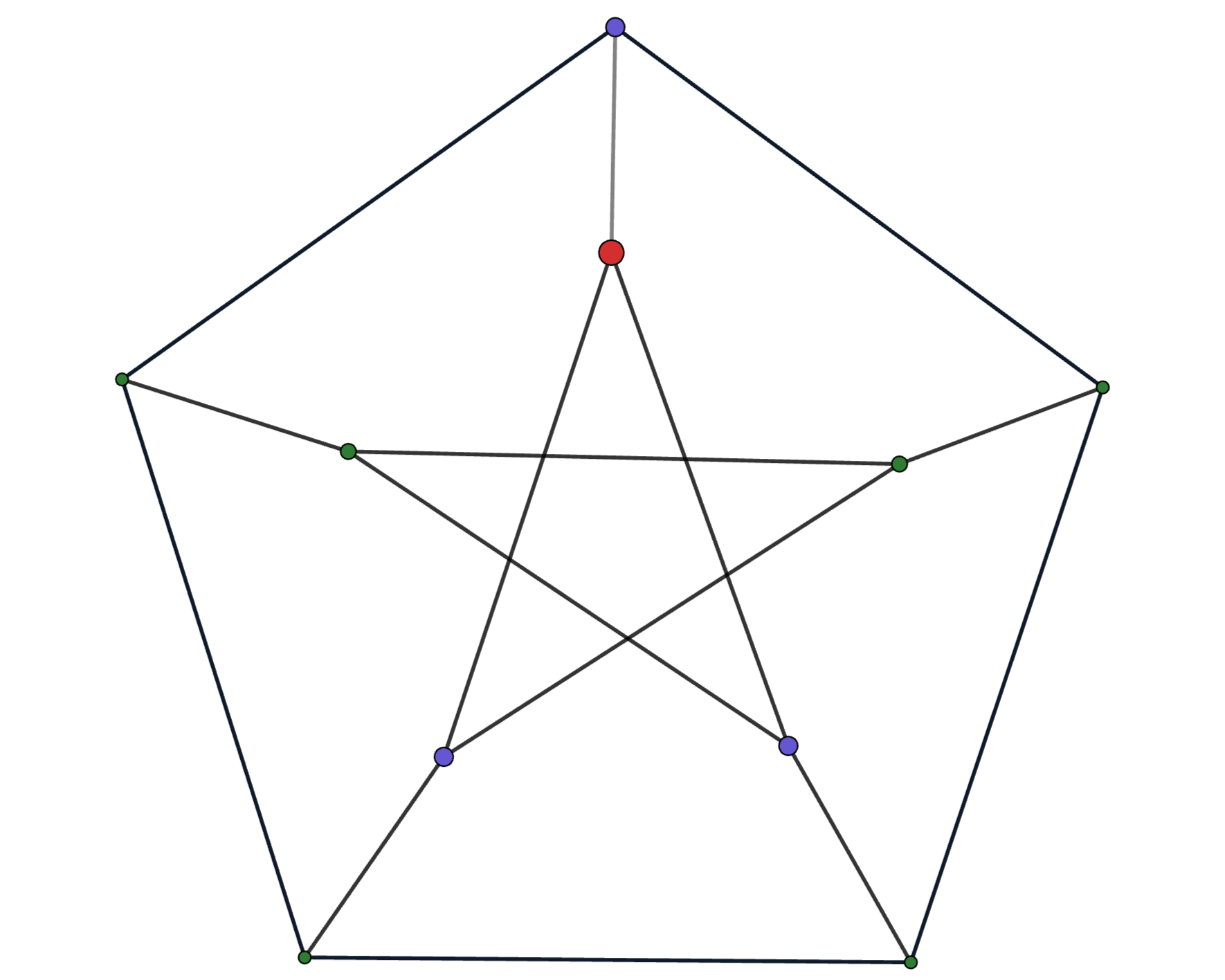}
	\end{minipage}%
	\begin{minipage}{.5\textwidth}
		\centering
		\includegraphics[width=0.7\linewidth]{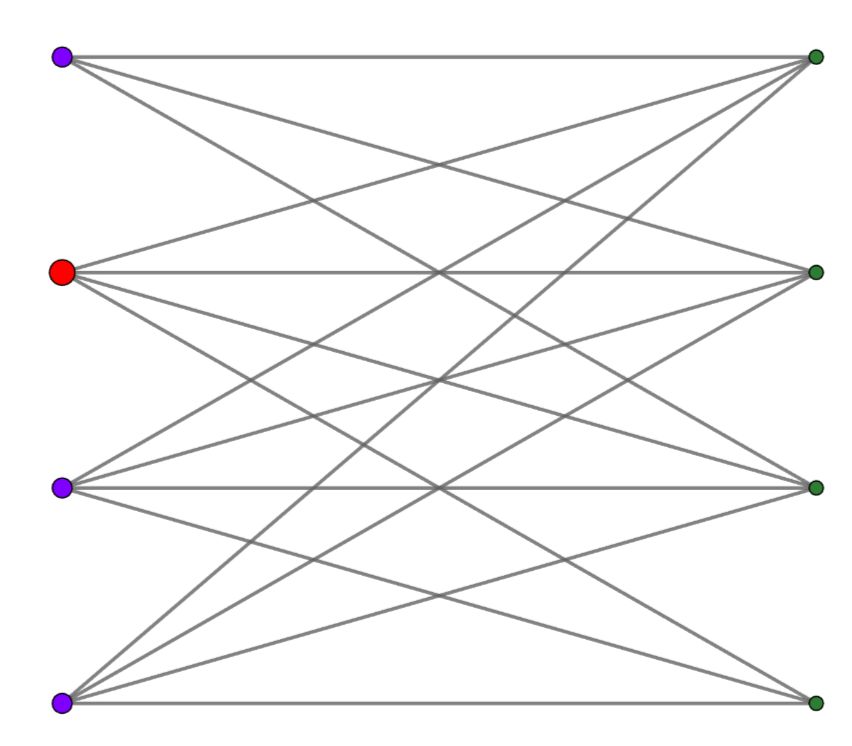}
	\end{minipage}
	\caption{Diagram showing the weight distribution of the continuous-time quantum walk on graphs of the form $\Gamma_0 \mathop\square  G_F$ for $G_F$ as Petersen graph (left) and the complete bipartite graph $K_{4,4}$(right). In both the figures, the walk starts at the red vertex, and it is the vertex where most of the weight concentrates. The purple vertices have the second-highest weight and the green vertices, the lowest.}
	\label{fig:Pet&Bipart}
\end{figure}

For complete bipartite graphs $K_{m,n}$, the vertices can be divided into three different classes sharing the same weights on them. Similar to the case of complete graphs, most of the mass gets concentrated on the starting vertex. The vertices non-adjacent to the starting vertex form the second group by weight. Finally, the vertices adjacent to the starting vertex make up the third group, following in weight distribution, having the least amount of weight. For the graph $K_{m,n}$, we denote the function $d$ by $d_{m,n}$. Then we have the values of $d_{n,n}$ as $d_{4,4}(p,p) \approx 0.59, d_{5,5}(p,p)=0.66,  d_{6,6}(p,p)\approx 0.71$ and $d_{100,100}(p,p) \approx 0.98$.

\subsubsection{Petersen graph} The Petersen graph is a popular 3-regular non-planar graph on 10 vertices, which is also the complement of the line graph of the complete graph $K_5$. For the Petersen graph (see Figure \ref{fig:Pet&Bipart}), $d(p,p)=21/50$ for all $p$, $d(p,q)=49/450$ for all $q$ adjacent to $p$ and $d(p,q)=19/450$ for all other vertices $q$. 

%\begin{figure*}[t!]
%	  \centering
%	\begin{subfigure}
%		\centering
%	\includegraphics[width=0.4\linewidth]{Petersen.jpeg}
%	\caption{(a)}
%	\end{subfigure}%
%	\begin{subfigure}
%	\centering
%	\includegraphics[width=0.4\linewidth]{bipartite.jpg}
%	\caption{(b)}
%\end{subfigure}%
%
%\end{figure*}

\providecommand{\bysame}{\leavevmode\hbox to3em{\hrulefill}\thinspace}
\providecommand{\MR}{\relax\ifhmode\unskip\space\fi MR }
% \MRhref is called by the amsart/book/proc definition of \MR.
\providecommand{\MRhref}[2]{%
  %\href{http://www.ams.org/mathscinet-getitem?mr=#1}{#2}
}
\providecommand{\href}[2]{#2}

\end{document}